\documentclass[10pt,a4paper]{article}
 \usepackage[colorlinks=true]{hyperref}
\hypersetup{urlcolor=blue, citecolor=red}

\usepackage{amsmath, amssymb, amsfonts, amsthm}
\usepackage{mathrsfs}

\usepackage{enumerate}  

\usepackage{dsfont}    
\usepackage{url}       
\usepackage{color}     
\usepackage{tikz}      
\usetikzlibrary{positioning} 
\usepackage{wrapfig}   
\usepackage{nccmath}   

\makeatletter
\renewcommand*\env@matrix[1][c]{\hskip -\arraycolsep
  \let\@ifnextchar\new@ifnextchar
  \array{*\c@MaxMatrixCols #1}}
\makeatother

\makeatletter
  \newcommand{\eqnum}{\leavevmode\hfill\refstepcounter{equation}\textup{\tagform@{\theequation}}} 
\makeatother

\newcommand{\A}{\mathcal{A}}

\newcommand{\Fi}{\mathrm{Fi}}
\newcommand{\bw}{\mathrm{bw}}
\newcommand{\fw}{\mathrm{fw}}
\newcommand{\I}{\mathcal{I}}

\renewcommand{\L}{\mathcal{L}}

\newcommand{\NN}{\mathbb{N}}

\renewcommand{\P}{\mathcal{P}}
\newcommand{\PP}{\mathbb{P}}
\newcommand{\Q}{\ensuremath{\mathcal{Q}}}
\newcommand{\RR}{\mathbb{R}}
\newcommand{\R}{\mathcal{R}}
\newcommand{\sR}{\mathcal{R}}
\newcommand{\sY}{\mathcal{Y}}
\newcommand{\Y}{\mathcal{Y}}
\renewcommand{\S}{\mathcal{S}}
\newcommand{\tp}{^{\mathsf{T}}}

\newcommand{\ceq}{c^{\mathrm{eq}}}
\newcommand{\sym}{^{\mathrm{sym}}}
\newcommand{\asym}{^{\mathrm{asym}}}
\newcommand{\as}{_{\mathrm{a}}^{\mathrm{s}}}
\newcommand{\sa}{_{\mathrm{s}}^{\mathrm{a}}}
\newcommand{\rev}{^{\mathrm{rev}}}
\newcommand{\irr}{^{\mathrm{irr}}}
\newcommand{\pair}[3]{\theta_{#3}\!\left({#1},{#2}\right)\!}

\newcommand{\llgg}{\;\substack{{\scriptscriptstyle\ll}\\[-0.15em]{\scriptscriptstyle\gg}}\;}

\let\div\relax
\DeclareMathOperator*\div{\mathrm{div}}

\DeclareMathOperator\grad{\nabla\!}

\DeclareMathOperator\Prob{Prob}

\newcommand{\super}[1]{^{\scriptscriptstyle{(#1)}}}

\newtheorem{theorem}{Theorem}[section]

\newtheorem{proposition}[theorem]{Proposition}
\newtheorem{corollary}[theorem]{Corollary}

\newtheorem{assumption}[theorem]{Assumption}
\newtheorem{example}[theorem]{Example}

\newenvironment{remark}
  {\par\medbreak\refstepcounter{theorem}%
    \noindent\textbf{Remark~\thetheorem. }}%
  {\qed\par\medskip}



\numberwithin{equation}{section}

\title{Orthogonality of fluxes in general nonlinear reaction networks}
\author{D~ R.~Michiel Renger\thanks{renger@wias-berlin.de}\, and Johannes Zimmer\thanks{j.zimmer@bath.ac.uk}}

\begin{document}
\maketitle
	
\begin{abstract}
  We consider the chemical reaction networks and study currents in these systems. Reviewing recent decomposition of rate
  functionals from large deviation theory for Markov processes, we adapt these results for reaction networks.  In particular, we
  state a suitable generalisation of orthogonality of forces in these systems, and derive an inequality that bounds the free
  energy loss and Fisher information by the rate functional.
\end{abstract}

\section{Introduction}
\label{sec:introduction}

It is now becoming widely accepted that fluxes hold a key to understanding many phenomena in non-equilibrium
thermodynamics. Typically, non-equilibrium systems are forced out of equilibrium by external forces. Since such forces may cause
`divergence-free' fluxes, the forces can not be balanced by changes in mass densities unless one takes the full fluxes into
account. In fact, a lot of thermodynamic information can be extracted from microscopic fluctuations of fluxes on the
large-deviations scale; this idea is the basis of Macroscopic Fluctuation Theory (MFT)~\cite{Bertini2015a}.

Although the setting of our paper will be slightly different, we briefly describe the typical setting considered in MFT. One
considers a random particle system with particle density $C\super{n}(x,t)$ (we write $C$ instead of $\rho$ as usual for a density
to emphasise the analogy for what follows) and particle flux $J\super{n}(x,t)$ connected via the continuity equation
$\dot C\super{n}(t)= - \div J\super{n}(t)$, where $\div$ is a divergence operator. The parameter $n$ controls the number of
elements in the system. When this number is sent to infinity, this can lead --- in a suitable scaling of space and time -- to a
macroscopic limit:
\begin{align}
  &\dot c(t) = \Gamma j(t), \text{ with}
    \label{eq:flux-MFT}\\
  &j(t) = \kappa(c(t)),
    \label{eq:intro cont eq}
\end{align}
for some model-specific operator $\kappa$; here $\Gamma \xi = -\div \xi$ (we will consider similar structures on the level of
chemical networks below, see~\eqref{eq:chem-flux} below). The microscopic fluctuations around this macroscopic limit are often
characterised by a large-deviation principle. In the context considered here, this means that the probability of observing
trajectories which deviate from the macroscopic limit converges exponentially:
\begin{multline}
  \label{eq:quadratic ldp}
  \Prob\big( (c\super{n},j\super{n}) \approx (c,j)\big) \\ 
    \stackrel{n\to\infty}{\sim} \exp\!\left(-n \left[\I_0(c(0)) + 
    \int_0^T\!\lVert j(t)-\kappa(c(t))\rVert^2_{L^2(f(c(t))}\,dt\right]\right),
\end{multline}
for some model-specific function $f$; $\I_0$ is a functional depending on the initial data $(c\super{n}(0))_n$. For example, for
independent Brownian particles, $f(c)=c^{-1}$, and for the simple symmetric exclusion process $f(c)=c^{-1}(1-c)^{-1}$.

The corresponding inner product is a very powerful tool that can be used to define orthogonal decompositions
of fluxes into ``reversible'' and ``irreversible'' parts. With this decomposition the large-deviation rate function also
decomposes,
\begin{multline}
  \label{eq:quadratic split}
  \lVert j(t)-\kappa(c(t))\rVert^2_{L^2(f(c(t))} =\\ \lVert j(t)\rev-\kappa\rev(c(t))\rVert^2_{L^2(f(c(t))} 
  + \lVert j\irr(t)-\kappa\irr(c(t))\rVert^2_{L^2(f(c(t))}.
\end{multline}
Moreover, by duality one also obtains an inner product on (generalised) forces, allowing to orthogonally decompose forces into
reversible and irreversible parts.

In this work we consider a similar setting, however for a class of systems with a different noise so that the large deviations are
entropic rather than quadratic,
\begin{align}
  \label{eq:intro entropic RF}
  \Prob\big( (c\super{n},j\super{n}) \approx (c,j)\big) \stackrel{n\to\infty}{\sim}
  \exp\!\left(-n \left\lbrack\I_0(c(0)) + \int_0^T\!\S\big(j(t)\mid \kappa(c(t))\big)\,dt\right\rbrack\right),
\end{align}
see~\eqref{eq:entropic cost} below for the precise definition. To facilitate the intuitive picture, we interpret these systems as
models for chemical reactions, although the range of applicability is much wider. In particular, these models are more natural
than white-noise driven models when the underlying state space remains discrete. Their non-quadratic fluctuation costs do not
suggest a natural inner product, and so it is not evident how to meaningfully decompose the cost as in~\eqref{eq:quadratic split}
for the quadratic case.

Recently, a new concept of generalised orthogonality was introduced for entropic cost functions corresponding to independent
particles~\cite{Kaiser2018a}. In that case, the operator $\kappa$ is linear and the continuity operator will be the negative
discrete divergence $\Gamma=-\div$. In this work we extend this orthogonality concept to the more general case of randomly
inter(re-)acting particles. In Sections~\ref{sec:setting} and~\ref{sec:force structures}, we collect results from the literature,
adapting and extending them where needed to the context of chemical reaction networks. We therefore omit the proofs and sometimes
the precise technical assumptions, which can be found in the quoted literature. Sections~\ref{sec:orth}
and~\ref{sec:Fisher-bounds} give orthogonal decomposition results and bounds on the entropy and a term resembling the Fisher
information in the classical setting.

\section{Microscopic model, macroscopic limit and large deviations}
\label{sec:setting}

As mentioned we study a model for chemical reactions, although the mathematical structure is applicable to any Markov jump process
for which the large-deviation principle of the form~\eqref{eq:intro entropic RF} holds~\cite{Kaiser2018a,PattersonRenger2019}. Let
$\sY$ be the set of reactants involved, for example $\sY := \left\{ \mathrm{Na}, \mathrm{CL}_2, \mathrm{NaCL} \right\}$.  We
denote by $\sR$ the set of reactions, in the example
\begin{equation*}
  \sR :=
  \left\{2 \mathrm{Na} + \mathrm{CL}_2 \to 2\mathrm{NaCL}, \  2\mathrm{NaCL} \to 2 \mathrm{Na} + \mathrm{CL}_2 \right\}.
\end{equation*}
The stoichiometric information of a reaction $r$ is captured in the state change vector $\gamma\super{r}\in\RR^\sY$, which
describes with negative entries how many reactants of each species are consumed and with positive entries the creation of
products, for example $\gamma\super{r}:=(-2,-1,2)$ for the reaction $r \colon 2 \mathrm{Na} + \mathrm{CL}_2 \to 2\mathrm{NaCL}$.
We shall assume that $\sR$ is the disjoint union of the set of forward reactions $\sR_\fw$, in this example
$\sR_\fw:=\left\{2 \mathrm{Na} + \mathrm{CL}_2 \to 2\mathrm{NaCL}\right\}$, and the set of backward reactions $\sR_\bw$, such that
for each forward reaction $r\in\R_\fw$ there is exactly one backward reaction $\bw(r)\in\sR_\bw$ with:
\begin{equation}
  \label{eq:bw-r}
  \gamma^{\bw(r)}= - \gamma\super{r}.
\end{equation}
This assumption is no loss of generality, as we can always introduce phantom reactions with zero reaction rates, i.e., reactions
that exist on paper but do not take place. We sometimes also write $\bw(r)\in\R_\fw$ for $r\in\R_\bw$.  The state change vectors
are collected in the matrix $\Gamma:=\lbrack \gamma\super{r}\rbrack_{r\in\R}\in\RR^{\sY\times\sR}$. We deliberately use the same
notation as in~\eqref{eq:flux-MFT}, to emphasise the analogy of that equation to~\eqref{eq:chem-flux} below.


We now further specify the microscopic particle system that models such reaction network, and introduce the time-reversed
process. Next, we describe the macroscopic limit and the corresponding large deviations. We then briefly summarise in the
remainder of this section some results from~\cite{Renger2018a} which will be needed in the current work: time-reversal symmetries
and the relations between forward and backward reaction rates that can be derived.

\subsection{Microscopic model}
\label{subsec:micro model}

The microscopic model will be a Markov jump process consisting of randomly reacting particles in a large volume of size $V$, which
now controls number of particles in the system. We assume that there is no spatial dependence of the reactions. The
\emph{empirical concentration measure} is
\begin{equation*}
  C\super{V}_y(t) := \frac 1 V \#\big\{\text{particles of species } y \text{ present at time } t\big\}.
\end{equation*}
A reaction $r\in\R$ occurs randomly with concentration-dependent intensity $k\super{V}_r = k\super{V}_r(C\super{V}(t))$, upon
which the concentration is updated with $C\super{V}(t):=C\super{V}(t^-)+V^{-1}\gamma\super{r}$, where $C\super{V}(t^-)$ is the
limit from the left. These intensities $k\super{V}_r$ are also called the \emph{jump rates} or \emph{propensities}. In addition,
we would like to introduce a reaction flux $J\super{V}_r(t)$ that measures the amount of reactions taking place at time
$t$. However, since the process has jumps, it is mathematically easier to introduce the \emph{(time-)integrated flux},
\begin{equation*}
  W_r\super{V}(t) := \frac 1 V \# \left\{\text{reactions } r \text{ occurred in time } (0,t]\right\}.
\end{equation*}
The fluxes and concentrations are then related by the continuity equation $C\super{V}(t)=C\super{V}(0) + \Gamma W\super{V}(t)$, or
formally
\begin{equation}
  \label{eq:chem-flux}
  \dot C\super{V}(t)=\Gamma\dot W\super{V}(t)=\Gamma J\super{V}(t) . 
\end{equation}
We point out that this equation is the equivalent for chemical reactions to~\eqref{eq:flux-MFT}, for the MFT setting sketched in
the introduction. The pair $\left(C\super{V}(t), W\super{V}(t)\right)$ is now a Markov process in $\RR^\Y\times\RR^\R$; its
generator is given by
\begin{align*}
  (\Q\super{V}f)(c,w) := \sum_{r\in\R} k\super{V}_r(c)
  \left\lbrack f(c+\tfrac 1 V \gamma\super{r},w+\tfrac1V\mathds1_r) - f(c,w)\right\rbrack.
\end{align*}
We remark that these are the sole reasons for considering integrated fluxes: although the concentration $C\super{V}(t)$ is a
Markov process, $C\super{V}(t)$ paired with the non-integrated flux $J\super{V}(t)$ is not, and moreover, \eqref{eq:chem-flux}
only holds in a weak, measure-valued sense. However in the macroscopic regime $V\to\infty$, these issues do not play a role
anymore, so that later we can focus on the non-integrated fluxes.

We collect the needed technical assumptions on the Markov chain as follows.
\begin{assumption}
  \label{ass:nice process}
  Assume that for all $V$ the process $C\super{V}(t)$ is non-explosive on $(0,T)$ (that is, for almost all starting points, almost
  all trajectories do not exhibit infinitely many jumps in $(0,T)$) and remains almost surely within a compact set, typically a
  simplex that describes mass conservation. Furthermore, the process has a unique, coordinate-wise positive invariant measure
  $0<\pi\super{V}\in\P(\tfrac1V\NN_0^I)$. Integrated fluxes are assumed to satisfy $W\super{V}(0)=0$ almost surely in $V$.
\end{assumption}

\begin{example}
  \label{ex:CME}
  The typical example that is used to model microscopic chemical reactions is $k\super{V}_r(c)=\omega\super{r} V \prod_{y\in\Y}
  (1/V)^{\alpha\super{r}_y} \alpha\super{r}_y! {c_y V \choose \alpha\super{r}_y}$, where $\omega\super{r}$ is a reaction-specific
  constant and $\alpha_y\super{r}$ are the number of $y$-reactants consumed in reaction $r$. In that case the master equation for
  the Markov process $C\super{V}(t)$ is called the ``Chemical Master Equation'', and in many cases the invariant measure
  $\pi\super{V}$ is explicitly known and positive~\cite{Anderson2011a}. In general, explosion may occur for such models, but it
  can be ruled out by imposing an additional mass conservation assumption~\cite{PattersonRenger2019}.
\end{example}

\subsection{The time-reversed process}
\label{sec:Time-reversal}

Here we introduce the time-reversed process which will be crucial throughout the paper. For a path $(C\super{V}(t),W\super{V}(t))$
and the given final time $T$, we define the time-reversed path as
\begin{equation*}
  \left(\overleftarrow C\super{V}(t),\overleftarrow W\super{V}(t)\right)
  :=\left(C\super{V}(T-t),{W\super{V}}\tp(T) - {W\super{V}}\tp(T-t)\right),
\end{equation*}
where ${W\super{V}}\tp_{\bw(r)}(t):=W\super{V}_r(t)$; by construction, these time-reversed integrated fluxes
$\overleftarrow W\super{V}(t)$ remain non-negative, non-decreasing, and initially calibrated at $0$, as the original integrated
fluxes $\overleftarrow W\super{V}(t)$.

\begin{theorem}[{\cite[Proposition~4.1]{Renger2018a}}, time-reversal]
  Let $\PP\super{V}_{\Q\super{V},\pi\super{V}}$ be the path measure of the Markov process $(C\super{V}(t),W\super{V}(t))$ with
  generator $\Q\super{V}$ and initial distribution $\pi\super{V}\times\delta_0$. Then
  \begin{equation}
    \PP\super{V}_{\Q\super{V},\pi\super{V}}\big( (\overleftarrow C\super{V}(t),\overleftarrow W\super{V}(t)) \in \A\big) 
    = 
    \PP\super{V}_{\overleftarrow\Q\super{V},\pi\super{V}}\big( (C\super{V}(t),W\super{V}(t)) \in \A\big),
  \label{eq:time-reversal of path measures}
  \end{equation}
  where the reversed generator is
  \begin{align}
    \label{eq:reverse generator}
    (\overleftarrow\Q f)(w)
    = \sum_{r\in \R} \overleftarrow k\!\super{V}_{\bw(r)}(c)
    \big\lbrack f(c+\tfrac1V\gamma^{\bw(r)},w+\tfrac1V\mathds1_{\bw(r)})-f(c,w)\big\rbrack,
  \end{align}
  and the reversed rates are related to the forward rates through
  \begin{align}
    \label{eq:reversed rates}
    \overleftarrow k\!\super{V}_{\bw(r)}(c):=
    \frac{\pi\super{V}(c+\tfrac1V\gamma^{\bw(r)})}{\pi\super{V}(c)}
    k\super{V}_r(c+\tfrac1V\gamma^{\bw(r)}).
  \end{align}
\label{th:time reversal}
\end{theorem}

\begin{remark}
  A special role is played for systems that satisfy ``microscopic detailed balance'', i.e., for $V>0$:
\begin{align*}
  \pi\super{V}(c) k\super{V}_r(c) = \pi\super{V}(c+\tfrac1V\gamma\super{r}) k\super{V}_r(c+\tfrac1V\gamma\super{r})
    &&
  \text{for all } c\in\RR^\Y \text{ and } r\in\R.
\end{align*}
Indeed, by~\eqref{eq:reversed rates}, this is equivalent to $\overleftarrow k\!\super{V}_{r}(c) = k\super{V}_r(c)$, and hence to
the reversibility of the Markov process, i.e., $\overleftarrow{\Q}\super{V}=\Q\super{V}$ in~\eqref{eq:time-reversal of path
  measures} (not to be confused with thermodynamic reversibility). As we will see, on a macroscopic level this condition
corresponds to conservative forces, or the absence of external forces. Since we are mainly interesting in the behaviour of systems
undergoing external forces, we shall not assume this condition, but merely use it as a validity check at some places.
\label{rem:micro DB}
\end{remark}

\subsection{Macroscopic limit and large deviations}
\label{subsec:limit and ldp}

We now study the macroscopic behaviour of the system in the limit $V\to\infty$. To this aim we need to make the following
assumption.
\begin{assumption}
  Let the jump rates converge in average to reaction rates, that is, $\sup_c \lvert V^{-1}k\super{V}_r(c)-\kappa_r(c)\rvert\to 0$
  as $V\to\infty$. Furthermore, we assume that the rates $\kappa_r$ satisfy the assumptions of~\cite{PattersonRenger2019} allowing for
  a large-deviation principle; in particular, a rate $\kappa_r(c)$ has to vanish if the associated reaction would lead to negative
  concentrations, and the mapping $c\mapsto \kappa_r(c)$ must be non-decreasing, Lipschitz continuous,
  superhomogeneous~\cite[Assumption~2.2(vi)]{PattersonRenger2019} and bounded on the compact set of concentrations of
  Assumption~\ref{ass:nice process}.  We require the invariant distributions $\pi\super{V}\in\P(\tfrac1V\NN_0^\sY)$ to satisfy a
  large-deviation principle
  \begin{equation}
    \pi\super{V}(C\super{V}\approx c) \stackrel{V\to\infty}\sim \exp\big(-V\I_0(c)\big),
  \label{eq:I0}
  \end{equation}
  for some $\I_0\colon\RR^\sY_+\to\lbrack0,\infty\rbrack$. Finally, we assume that $\I_0$ is almost everywhere differentiable, as
  it holds for example for convex functionals.
\label{ass:kappa}
\end{assumption}

By a straightforward extension of Kurtz' Theorem~\cite{Kurtz1970a}, in the limit $V\to\infty$, the random path
$\big(C\super{V}(t),W\super{V}(t)\big)$ converges (in measure) to the deterministic solution of the coupled equations
\begin{align}
  \label{eq:kurtz}
  \dot c(t) = \Gamma \dot w(t),
   \text{ and }
  \dot w(t) = \kappa\big(c(t)\big),
\end{align}
which is the chemical reaction equivalent of~\eqref{eq:flux-MFT} and~\eqref{eq:intro cont eq}.

Since the noise is essentially a re-scaled Poissonian, a large-deviations principle holds, which we now recall
(see~\cite{PattersonRenger2019} for the precise statement). In the following theorem and throughout the paper we adapt the notation of
the relative entropy:
\begin{equation}
  \S(j\mid \kappa(c)):=\sum_{r\in\sR} j_r\log\mfrac{j_r}{\kappa_r(c)} - j_r + \kappa_r(c).
  \label{eq:entropic cost}
\end{equation}

\begin{theorem}[{\cite{PattersonRenger2019}}]
  As in Theorem~\ref{th:time reversal}, let $\PP\super{V}_{\Q\super{V},\pi\super{V}}$
  be the path measure of the Markov process $(C\super{V}(t),W\super{V}(t))$
  with generator $\Q\super{V}$
  and initial distribution $\pi\super{V}\times\delta_0$. Then for any $(c,w)\in W^{1,1}(0,T;\RR^\Y\times\RR^\R)$
  \begin{multline*}
    \PP\super{V}_{\Q\super{V},\pi\super{V}}\!\left( (C\super{V},W\super{V}) \approx (c,w)\right)
    \\ \stackrel{V\to\infty}{\sim}
    \exp\left(-V\!\left[\I_0(c(0)) + \int_0^T\!\S\big(\dot w(t)\mid\kappa(c(t))\big)\,dt\right]\right).
  \end{multline*}
  \label{th:ldp}
\end{theorem}

The non-negative cost functional $\int_0^T\!\S\big(\dot w(t)\mid\kappa(c(t))\big)\,dt$ can be interpreted as the free energy
required to force the system to follow a given path, $(c,w)$ rather than the path that solves the macroscopic
equation~\eqref{eq:kurtz}, whereas $\I_0$ measures the cost of deviating from the given initial data. This interpretation goes
back to~\cite{Onsager1953a}; see also~\cite{Mielke2017a} for a more detailed account.

\begin{example}
  For the Chemical Master Equation discussed in Example~\ref{ex:CME}, the corresponding reaction rates are
  $\kappa\super{r}(c)=\omega\super{r}\prod_{y\in\Y} c_y^{\alpha\super{r}_y}$. This specific form of the reaction rates is called
  \emph{mass-action kinetics}~\cite{Kurtz1970a, Anderson2011a}. Moreover, in that case $\I_0(c):=\S(c\mid \ceq)$, where $\ceq$ is
  the equilibrium concentration $0=\Gamma\kappa(\ceq)$~\cite{Mielke2017a}.
\end{example}

\subsection{Time reversal symmetries and implications}
\label{subsec:time reversal symmetries}

The large-deviation result Theorem~\ref{th:ldp} also applies to the time-reversed process
$(\overleftarrow C\super{V}(t),\overleftarrow W\super{V}(t))$. Combining this with the time reversal result of
Theorem~\ref{th:time reversal} yields the following time-reversal symmetry:
\begin{corollary}
  For any path $(c,w)\in W^{1,1}(0,T;\RR^\Y\times\RR^\R)$, it holds
  \begin{equation}
    \label{eq:time-reversal ldp}
    \I_0(c(0)) + \int_0^T\!\S\big(\dot w(t)\mid\kappa(c(t))\big)\,dt = \I_0(c(T)) + \int_0^T\!\S\big(\dot w\tp(t)\mid 
    \overleftarrow\kappa(c(t))\big)\,dt,
  \end{equation}
  or (using $\dot c=\Gamma\dot w$), for a dense set of $(c,j)=(c,\dot w)\in\RR^\Y\times\RR^\R$ for which the chain rule holds,
  \begin{equation}
    \label{eq:time-reversal}
    \S(j\mid \kappa(c))) - \S(j\tp\mid\overleftarrow\kappa(c)) = \grad\I_0(c)\cdot \Gamma j.
  \end{equation}
\end{corollary}

Although the reversed propensities $\overleftarrow k\!\super{V}_r(c)$ are explicitly known by~\eqref{eq:reversed rates} \emph{if}
the invariant measure $\pi\super{V}$ is known, the reversed reaction rates $\overleftarrow \kappa\!_r(c)$ may generally not be
explicit. However, using the fact that~\eqref{eq:time-reversal} holds for all $j$, one finds the
relations~\cite[Section~4.2]{Renger2018a}, if the quantities are defined,
\begin{align}
  \grad\I_0(c)\cdot\gamma\super{r} &= \log\frac{\overleftarrow\kappa\!_{\bw(r)}(c)}{\kappa_r(c)},
    \label{eq:FD}\\
  \sum_{r\in\R}\kappa_r(c) &= \sum_{r\in\R} \overleftarrow\kappa\!_{r}(c),
    \label{eq:sum kfw is sum kbw}\\
  \kappa_r(c)\kappa_{\bw(r)}(c) &= \overleftarrow\kappa\!_r(c)\overleftarrow\kappa\!_{\bw(r)}(c).
    \label{eq:fw mob is bw mob}
\end{align}
Note in particular that the ratio $\overleftarrow\kappa\!_{\bw(r)}(c)/\kappa_r(c)$ is well-defined whenever $\grad\I_0(c)$ is
well-defined, i.e., on a dense subset.  We mention that~\eqref{eq:fw mob is bw mob} was not stated in~\cite{Renger2018a} but
follows directly from~\eqref{eq:FD} and the antisymmetry~\eqref{eq:bw-r}.

\begin{remark}
  \label{rem:chem DB}
  Building further upon Remark~\ref{rem:micro DB}, if microscopic detailed balance holds, then this is also true for the reaction
  rates: $\overleftarrow\kappa(c)=\kappa(c)$, which is often called \emph{chemical detailed balance}, or in case of mass-action
  kinetics, the \emph{Wegscheider condition}.
\end{remark}

\section{Force structures}
\label{sec:force structures}

In this section, we introduce force structures; these can be seen as a non-equilibrium generalisation of gradient flows. Such
structure does in general not exist unless we consider net fluxes. Therefore, we first introduce net fluxes, then force
structures, and we finally discuss the force structure that corresponds to the reversed dynamics.

\subsection{Net fluxes: macroscopic limit and large deviations}
\label{subsec:net fluxes}

Since reactions are ordered in forward-backward pairs with opposite state change vectors~\eqref{eq:bw-r}, one can also introduce
the net integrated fluxes as $\bar{W}\super{V}(t):=W\super{V}(t) - {W\super{V}}(t)\tp$, i.e.,
\begin{equation}
  \label{eq:net flux}
  \bar{W}\super{V}_r(t):=W\super{V}_r(t) - W\super{V}_{\bw(r)}(t), \qquad\text{for } r\in \R_\fw.
\end{equation}
Keeping in mind that these net fluxes are only defined for $r\in\R_\fw$, we can use the same notation for the continuity equation,
cf.~\eqref{eq:chem-flux}:
\begin{align*}
  \dot C\super{V}(t)&=\Gamma\dot W\super{V}(t)=\sum_{r\in\R_\fw} \gamma\super{r} \dot W\super{V}_r(t) + \gamma^{\bw(r)} 
                      \dot W\super{V}_{\bw(r)}(t)\\
                    &=\sum_{r\in\R_\fw}\gamma\super{r}\dot{\bar W}\super{V}_r(t)=\Gamma \dot{\bar W}\super{V}(t).
\end{align*}
Again, by Kurtz' Theorem~\cite{Kurtz1970a}, the pair $(C\super{V}(t),\bar W\super{V}(t))$ converges in the macroscopic limit as
$V\to\infty$ to
\begin{equation}
  \label{eq:kurtz net}
  \dot c(t) = \Gamma \dot{\bar{w}}(t), \text{ and } \dot{\bar{w}}_r(t) = \kappa_r\big(c(t)\big) - \kappa_{\bw(r)}\big(c(t)\big), 
\end{equation}
and, by a contraction principle also satisfies the following large-deviation principle.

\begin{corollary}
  \label{cor:net ldp}
  \begin{multline*}
    \PP\super{V}_{\Q\super{V},\pi\super{V}}\!\left( (C\super{V}(t),\bar W\super{V}(t)) \approx (c,\bar w)\right)
    \\ \stackrel{V\to\infty}{\sim}
    \exp\left(-V\!\left[\I_0(c) + \int_0^T\!\L\big(c(t),\dot{\bar{w}}(t)\big)\,dt\right]\right),
  \end{multline*}
  where for $(c,\bar\jmath) \in\RR^\Y\times\RR^{\R_\fw}$:
  \begin{align}
    \L(c,\bar\jmath)&:= \inf_{\substack{j:(0,T)\to\RR^{\R}:\\ \bar\jmath = j - j\tp }} \S(j\mid\kappa(c)),\label{eq:net ldp cost}\\
                    &=\sum_{r\in\R_\fw} j_r\log\mfrac{j_r}{\kappa_r(c)} - j_r +\kappa_r(c) 
                      + (j_r-\bar\jmath_r)\log\mfrac{j_r-\bar\jmath_r}{\kappa_{\bw(r)}} - (j_r-\bar\jmath_r) + \kappa_{\bw(r)}, \notag\\
                    &\hspace{5cm} \text{where } j_r:=\tfrac12\bar\jmath_r + \sqrt{\tfrac14\bar\jmath_r^2 + \kappa_r(c)\kappa_{\bw(r)}(c)}. 
                      \notag
  \end{align}
\end{corollary}

\subsection{Force structure of the cost function}
\label{sec:Decomp-force-struct}

Following~\cite{Schnakenberg1976a, Maes2008a, Maes2017a, Kaiser2018a, Renger2018a}, we now rewrite the cost function $\L$ in terms
of forces, affinities and energies. As will be clear from the formulas, these forces are only defined for concentrations that
satisfy the following condition:
\begin{equation}
  \label{q:weak DB}
  \kappa(c) \llgg \kappa_\bw(c) :\!\!\iff \left(  \kappa_r(c)=0\iff \kappa_{\bw(r)}(c)=0 \qquad \text{for all } r\in\R \right).
\end{equation}
Systems for which this condition holds for \emph{any} concentration $c$ are sometimes said to be in \emph{weak detailed balance}.

For any $c$ for which~\eqref{q:weak DB} holds, the thermodynamic forces are defined as (setting $0/0=1$)
\begin{equation}
  \label{eq:force field}
  F_r(c):=\mfrac12\log\frac{\kappa_r(c)}{\kappa_{\bw(r)}(c)},\qquad \text{ for }c \text{ satisfying}~\eqref{q:weak DB}\text{ and }
  r\in\R_\fw;
\end{equation}
these are the affinities of Schnakenberg~\cite{Schnakenberg1976a}. We also define the dual and primal dissipation potentials
(which are Legendre duals of each other):
\begin{align}
  \Phi^*(c,\zeta)    &:= 2\sum_{r\in\R_\fw}\sqrt{\kappa_r(c)\kappa_{\bw(r)}(c)}\big(\cosh(\zeta_r)-1\big),  \text{ and} \label{eq:Phi*}\\
  \Phi(c,\bar\jmath) &:=\sup_{\zeta\in \RR^{\R_\fw}} \zeta\cdot\bar\jmath - \Phi^*(c,\zeta) \notag\\
                     &=\sum_{r\in\R_\fw}2\sqrt{\kappa_r(c)\kappa_{\bw(r)}(c)}\!\left( 
                       \cosh^*\big(\tfrac{\bar\jmath_r}{2\sqrt{\kappa_r(c)\kappa_{\bw(r)}(c)}}\big) + 1\right).\notag
\end{align}

We can then decompose the cost function $\L$ as follows.
\begin{theorem}[{\cite{Maes2008a,Maes2017a,Kaiser2018a,Renger2018a}}]
  For any $c\in\RR_+^\sY$ for which $\kappa(c) \llgg \kappa_{\bw}(c)$ and any $\bar\jmath\in\RR^{\R_\fw}$,
  \begin{align}
    \label{eq:L=force structure}
    \L(c,\bar\jmath) = \Phi(c,\bar\jmath) + \Phi^*\big(c,F(c)\big) - F(c)\cdot\bar\jmath,
  \end{align}
  and this choice of $\Phi^*,\Phi$ and $F$ is unique (assuming $\Phi^*(c,0)=0=\Phi(c,0)$ and $\Phi(c,\cdot),\Phi^*(c,\cdot)$ are
  Legendre duals of each other~\cite[Proposition~2.1]{Mielke2014a}).
\end{theorem}
 
This decomposition has the following physical interpretation. The solution to~\eqref{eq:kurtz net} is the zero-cost flow
$\L(c(t),\bar\jmath(t))=0$, which is equivalent to
\begin{equation}
  \label{eq:nonlinear response}
  \bar\jmath(t) = \grad_\zeta\Phi^*\big(c(t),F(c(t))\big).
\end{equation}
Since $\Phi^*$ is not quadratic, this represents a non-linear response between forces and fluxes. Indeed, this also means that the
force $F$ must be scaled to be dimensionless; see~\cite{Mielke2017a} for the correct scaling with physical constants. Of
particular interest is the \emph{Fisher information}:
\begin{align*}
  \Phi^*\big(c,F(c)\big) &= \sum_{r\in \R_\fw} \big(\sqrt{k_
                           r(c)} - \sqrt{k_{\bw(r)}(c)}\big)^2.
\end{align*}
We interpret this as a Fisher information since for the case of independent Brownian particles
$\Phi^*(c,F(c))=\int_{\RR^d}\!\lvert\grad\sqrt{c(x)}\rvert^2\,dx$, see also~\cite{Adams2013a,Mielke2014a}; the expression above
has the same form, but with a discrete gradient and non-linear rates. The Fisher information will be discussed in more detail in
Section~\ref{sec:Fisher-bounds}.

\begin{remark}
  \label{rem:rev force}
  If chemical detailed balance $\overleftarrow \kappa=\kappa$ holds, see Remark~\ref{rem:chem DB}, then from~\eqref{eq:force
    field} and~\eqref{eq:FD} we find that the last term in~\eqref{eq:L=force structure} reads:
  \begin{equation*}
    -F(c)\cdot\bar\jmath = \tfrac12\Gamma\tp\grad\I_0(c)\cdot\bar\jmath = \tfrac12\grad\I_0(c)\cdot\dot c,
  \end{equation*}
  so that the force is indeed conservative as anticipated in Remark~\ref{rem:micro DB}. In this case,~\eqref{eq:nonlinear
    response} describes a nonlinear gradient flow, either in the space of concentrations~\cite{Mielke2017a} or in the space of
  integrated net fluxes~\cite{Renger2018b}.
\end{remark}

\subsection{Force structure for the reversed dynamics}
\label{sec:Force-struct-reverse}

We can also give a similar decomposition applied to the reversed dynamics, explained in Sections~\ref{sec:Time-reversal}
and~\ref{subsec:time reversal symmetries}. First, we apply Corollary~\ref{cor:net ldp} to the reversed dynamics,
using~\eqref{eq:time-reversal of path measures}:
\begin{corollary}
  \label{cor:reversed net ldp}
  Set $\overleftarrow{\bar W}\super{V}(t) := -\bar W\super{V}(T) + \bar W\super{V}(T-t)$, cf.~\eqref{eq:net flux}. Then:
  \begin{multline*}
    \PP\super{V}_{\Q\super{V},\pi\super{V}}\!\left( (\overleftarrow C\super{V}(t),\overleftarrow{\bar W}\super{V}(t)) \approx (c,\bar w)\right)
    \\ \stackrel{V\to\infty}{\sim}
    \exp\left(-V\!\left[\I_0(c) + \int_0^T\!\overleftarrow\L\big(c(t),\dot{\bar{w}}(t)\big)\,dt\right]\right),
  \end{multline*}
  where for $(c,\bar\jmath) \in\RR^\Y\times\RR^{\R_\fw}$:
  \begin{align}
    \overleftarrow\L(c,\bar\jmath)&:= \inf_{\substack{j:(0,T)\to\RR^{\R}:\\ \bar\jmath = j - j\tp }} \S(j\mid\overleftarrow\kappa(c)).
  \end{align}
\end{corollary}

We may then also decompose the reverse cost function $\overleftarrow\L$ as a force structure.
\begin{corollary}[\cite{Renger2018a}]
  For any $c\in\RR_+^\sY$ for which $\overleftarrow\kappa\!(c) \llgg \overleftarrow\kappa\!_{\bw}(c)$ and any
  $\bar\jmath\in\RR^{\R_\fw}$,
  \begin{align*}
    \overleftarrow\L(c,\bar\jmath) = \Phi(c,\bar\jmath) + \Phi^*\big(c,\overleftarrow F(c)\big) - \overleftarrow F(c)\cdot\bar\jmath,
  \end{align*}
  where $\Phi,\Phi^*$ are given by~\eqref{eq:Phi*}, and
  \begin{align*}
    \overleftarrow F\!_r(c)=\mfrac12\log\frac{\overleftarrow \kappa\!_r(c)}{\overleftarrow \kappa\!_{\bw(r)}(c)},
  \end{align*}
  and this choice of $\Phi^*,\Phi$ and $\overleftarrow F$ is unique.
\end{corollary}

Observe that the potentials $\Phi,\Phi^*$ are the same for the reversed dynamics due to~\eqref{eq:fw mob is bw mob}. In fact also
for the Fisher information $\Phi^*\big(c,F(c)\big) =\Phi^*\big(c,\overleftarrow F\!(c)\big)$ due to~\eqref{eq:sum kfw is sum kbw}
and~\eqref{eq:fw mob is bw mob}. 


\section{Orthogonal decomposition}
\label{sec:orth}

We now mimic the arguments from~\cite{Kaiser2018a} for the case of general nonlinear reactions. Throughout this section, we shall
write forces $F(c)$ and gradients $\grad\I_0(c)$ under the standing assumption that the quantities are well-defined and
sufficiently regular. These assumptions will be relaxed in Section~\ref{sec:Fisher-bounds}. First we explain how the dual
dissipation potentials can be decomposed; next we discuss how the forces can be decomposed accordingly.

\subsection{Orthogonal decomposition of the dual dissipation potentials}
\label{sec:Orth-decomp-dual}

The idea of this short subsection is to decompose $\Phi^*(c,\xi+\zeta)$ in a similar fashion as one would do in a Hilbert space:
$\tfrac12\lVert\xi+\zeta\rVert^2 = \tfrac12\lVert \xi\rVert^2 + \langle\xi,\zeta\rangle + \tfrac12\lVert \zeta\rVert^2$, see for
example~\cite[Section~IIC]{Bertini2015a}. However, since $\Psi^*$ is not quadratic, the construction is a bit more involved: the
pairing between forces becomes \emph{nonlinear}, and one of the dual dissipation potentials must be modified:
\begin{align*}
  \pair{\xi}{\zeta}{c} &:= 4\sum_{r\in\R_\fw} \sqrt{\kappa_r(c) \kappa_{\bw(r)}(c)} \sinh(\xi_r) \sinh(\zeta_r),\notag\\
  \Phi^*_\zeta(c,\xi)  &:= 2\sum_{r\in\R_\fw}\sqrt{\kappa_r(c)\kappa_{\bw(r)}(c)}\cosh(\zeta_r)\big(\cosh(\xi_r)-1\big).
\end{align*}

We can now write the decomposition as follows.
\begin{proposition} 
  \label{prop:decomposed Phi*}
  For any $c\in\RR^\Y$ and $\xi,\zeta\in\RR^\R$,
  \begin{align*}
    \Phi^*(c,\xi+\zeta)
    &= \Phi^*_\zeta(c,\xi) + \pair{\xi}{\zeta}{c} + \Phi^*(c,\zeta) \\
    &= \Phi^*(c,\xi) + \pair{\xi}{\zeta}{c} + \Phi^*_\xi(c,\zeta), \text{ and}\\
    \pair{\xi}{\zeta}{c}
    &=\Phi^*(c,\xi+\zeta) - \Phi^*(c,\xi-\zeta).
  \end{align*}
\end{proposition}

\begin{proof}
Both statements follow from the identity 
\begin{equation*}
  \cosh(\xi+\zeta) = \cosh(\xi)\cosh(\zeta) + \sinh(\xi)\sinh(\zeta).
\end{equation*}
\end{proof}

We stress that there is the choice which one of the two dissipation potentials on the right-hand side is modified. By a slight
abuse of notation, we shall also write $\Phi^*_F(c,\xi)=\Phi^*_{F(c)}(c,\xi)$ for a force field $F$.

\subsection{Orthogonal decomposition of the forces}
\label{sec:Orth-decomp-forc}

The force and its reversed counterpart are connected via the relations~\eqref{eq:FD}, \eqref{eq:sum kfw is sum kbw}, \eqref{eq:fw
  mob is bw mob}, see~\cite[Section~4.6]{Renger2018a}:
\begin{align}
  \label{eq:force field related to entropy production}
  F_r(c)+{\overleftarrow F\!}_{r}(c) 
  = -\big(\Gamma^\mathsf{T} \grad\I_0(c)\big)_r,
\end{align}
which is twice the force that one would have under chemical detailed balance, see Remark~\ref{rem:rev force}. This motivates the
splitting of general forces into a symmetric and an antisymmetric part $F(c) = F\sym(c) + F\asym(c)$: 
\begin{align*}
  F\sym_r(c) &:= \tfrac12\big(F_r(c) + \overleftarrow F\!_r(c)\big)
  = \tfrac12\log\frac{\kappa_r(c)}{\overleftarrow{\kappa}\!_{\bw(r)}(c)} = -\tfrac12(\Gamma\tp \grad\I_0(c))_r \qquad\text{and} \\
  F\asym_r(c)&:= \tfrac12\big(F_r(c) - \overleftarrow F\!_r(c)\big) = 
               \tfrac12\log\frac{\overleftarrow{\kappa}\!_{\bw(r)}(c)}{\kappa_{\bw(r)}(c)},
\end{align*}
where the latter indeed vanishes if chemical detailed balance $\kappa=\overleftarrow\kappa$ holds.  More generally, the
antisymmetric force measures how far the system is from chemical detailed balance. We briefly note that in case of linear
reactions, the antisymmetric force is a constant, independent of the concentration~\cite{Kaiser2018a}, which is no longer true for
nonlinear reactions.

The two force fields $F\sym, F\asym$ are indeed `orthogonal' to each other if we use the nonlinear pairing
$\pair{c}{\cdot}{\cdot}$ introduced above, as the next statement shows. 

\begin{proposition}[Generalisation of {\cite[Lemma~1]{Kaiser2018a}}]
  For any $c\in\RR_+^\sY$ for which $F\sym(c),F\asym(c)$ are well-defined,
  \begin{equation*}
    \pair{F\sym(c)}{F\asym(c)}{c} =0.
  \end{equation*}
\end{proposition}

\begin{proof}
  This follows from~\eqref{eq:sum kfw is sum kbw} and~\eqref{eq:fw mob is bw mob}.
\end{proof}

As a consequence of this and Proposition~\ref{prop:decomposed Phi*}, the Fisher information can be decomposed as follows:
\begin{corollary}[Generalisation of {\cite[Lemma~2 and Corollary~4]{Kaiser2018a}}]
  \label{cor:orth split L}
  For any $c\in\RR_+^\sY$ for which $F\sym(c),F\asym(c)$ are well-defined,
  \begin{align*}
    \Phi^*\big(c,F\sym(c)+F\asym(c)\big) &= \Phi^*_{F\asym}\big(c,F\sym(c)\big) + \Phi^*\big(c,F\asym(c)\big) \\
                                         &= \Phi^*\big(c,F\sym(c)\big) + \Phi^*_{F\sym}\big(c,F\asym(c)\big),
  \end{align*}
  and hence for any $c\in\RR_+^\sY$ for which $\kappa_r(c) \llgg \kappa_{\bw(r)}(c)$ and any $\bar\jmath\in\RR^{\R_\fw}$,
  \begin{align}
    \L(c,\bar\jmath)
    &= \Phi(c,\bar\jmath) + \Phi^*\big(c,F\asym(c)\big) - F\asym(c)\cdot \bar\jmath \notag\\
    &\hspace{4cm}       + \Phi^*_{F\asym}\big(c,F\sym(c)\big) +\tfrac12\grad\I_0(c)\cdot\Gamma\bar\jmath \label{eq:orth split L sym}\\
    &= \Phi(c,\bar\jmath) + \Phi^*\big(c,F\sym(c)\big) +\tfrac12\grad\I_0(c)\cdot\Gamma\bar\jmath \notag\\
    &\hspace{4cm}       + \Phi^*_{F\sym}\big(c,F\asym(c)\big) - F\asym(c)\cdot \bar\jmath. \label{eq:orth split L asym}
  \end{align}
\end{corollary}

\section{Fisher and entropy bounds}
\label{sec:Fisher-bounds}

Observe that the forces $F\sym(c),F\asym(c)$ are well-defined precisely if the conditions
$\kappa(c)\llgg\overleftarrow\kappa\!_{\bw}(c)\llgg\kappa_{\bw}(c)$ hold. It turns out that many terms in~\eqref{eq:orth split L sym}
and~\eqref{eq:orth split L asym} remain well-defined even when this condition is violated, see Example~\ref{ex:ill-posed
  forces}. More precisely, we may introduce the following notation, where the equalities on the right are true if
$\kappa_r(c)\llgg\kappa_{\bw(r)}(c)$:
\begin{align*}
  \Fi\as(c):=&\mfrac12\sum_{r\in\R_\fw\cup\R_\bw}
    \left( \sqrt{\kappa_r(c)} -\sqrt{\overleftarrow \kappa\!_{\bw(r)}(c)} \right)^2
    =
    \Phi^*_{F\asym}\big(c,F\sym(c)\big), \text{ and}\\
  \Fi\sa(c):=&\mfrac12\sum_{r\in\R_\fw\cup\R_\bw}
    \left(  \sqrt{\kappa_r(c)} -\sqrt{\overleftarrow \kappa\!_r(c)}\right)^2
    =
    \Phi^*_{F\sym}\big(c,F\asym(c)\big).
\end{align*}

\begin{remark}
  If we interpret the differences in as abstract gradients, then both quantities are of the form
  $\tfrac12\sum_r(\grad\sqrt{\kappa(c)})^2$, which coincides with the usual continuous-space Fisher information
  $\tfrac12\int\!\lvert\grad\sqrt{c(x)}\rvert^2\,dx=\tfrac12\int\!\tfrac{\lvert\grad c(x)\rvert^2}{c(x)}\,dx$. Recalling
  Remark~\ref{rem:chem DB}, the second quantity $\Fi\sa(c)$ measures how far the system is from being in detailed balance.
\end{remark}

\begin{example}
  \label{ex:ill-posed forces}
  Consider a two-state linear network with $\kappa_{xy}(c):=a_{xy}c_x/c_0^\mathrm{eq}$ and equilibrium concentration
  $c^\mathrm{eq}>0$, see also~\cite{Kaiser2018a}. The time-reversed rates are then given by
  $\overleftarrow\kappa_{xy}(c):=a_{yx}c_x/c^\mathrm{eq}_x$. Then the force is $F_{xy}(c)=\tfrac12\log\frac{a_{xy} c_x
    c^\mathrm{eq}_y}{a_{yx} c_y c^\mathrm{eq}_x}$, which decomposes into $F\sym_{xy}(c)=\tfrac12\log\frac{c_x c^\mathrm{eq}_y}{c_y
    c^\mathrm{eq}_x}$ and $F\asym_{xy}(c)=\tfrac12\log\frac{a_{xy}}{a_{yx}}$, and the two Fisher informations are
  $\Fi\as(c)=\tfrac12(a_{12}+a_{21})(\sqrt{c_1/c_1^\mathrm{eq}} - \sqrt{c_2/c_2^\mathrm{eq}})^2$ and
  $\Fi\sa(c)=\tfrac12(c_1/c_1^\mathrm{eq} + c_2/c_2^\mathrm{eq})(\sqrt{a_{12}}-\sqrt{a_{21}})^2$.  We thus see that the forces
  $F(c),F\sym(c)$ are not well-defined on the boundary $c=(1,0),(0,1)$, but the Fisher informations are.
\end{example}

The following proposition then generalises Corollary~\ref{cor:orth split L} to all concentrations, which possibly violate
$\kappa(c)\llgg\overleftarrow\kappa\!_{\bw}(c)\llgg\kappa_{\bw}(c)$. The resulting inequality, bounding Fisher information and
entropy by the rate functional, is known as a FIR inequality~\cite{HilderPeletierSharmaTse2019}. The inequality we prove here corresponds to the
inequality in that paper when choosing their parameter $\lambda=1/2$. Apart from the different proof strategy, our method shows
that the gap is precisely quantified by~\eqref{eq:orth split L sym}, at least for paths that stay away from the boundary,

\begin{proposition} 
  \label{prop:FIR}
  Assume that $\kappa$ satisfies the conditions of~\cite{PattersonRenger2019} needed for the large-deviations Theorem~\ref{th:ldp} to
  hold. Let $\I_0$ be lower semicontinuous. Then for any path $(c,\bar\jmath)\in W^{1,1}(0,T;\RR^\Y_+)\times
  L^1(0,T;\RR^{\R_\fw})$ such that $\I_0$ is continuous in a neighbourhood of $c(0)$,
  \begin{equation}
    \label{eq:FIR}
    \int_0^T\!\L\big(c(t),\bar\jmath(t)\big)\,dt \geq \int_0^T\!\Fi\as\big(c(t)\big)\,dt + \tfrac12\I_0\big(c(T)\big) - 
    \tfrac12\I_0\big(c(0)\big).
  \end{equation}
\end{proposition}

\begin{proof}
We use two approximation steps. In the first step we consider $(c,\bar\jmath)\in\RR_+^\Y\times\RR^\R$ for which only
$\kappa(c)\llgg\overleftarrow\kappa\!_{\bw}(c)$ holds. While keeping this pair fixed, we define a new reaction network with
propensities and corresponding reaction rates:
\begin{align*}
  k\super{V,\epsilon}_r(c):=k\super{V}_r(c) + \epsilon\overleftarrow k\!\super{V}_r(c),
  &&\text{and}&&
  \kappa\super{\epsilon}_r(c):=\kappa_r(c) + \epsilon\overleftarrow\kappa\!_r(c).
\end{align*}
Then $\kappa\super{\epsilon}(c)\llgg\overleftarrow\kappa\!\super{\epsilon}_{\bw}(c)\llgg\kappa\super{\epsilon}_{\bw}(c)$,
so~\eqref{eq:orth split L sym} holds. By convex duality~$\Phi(c,\bar\jmath) + \Phi^*\big(c,F\asym(c)\big) - F\asym(c)\cdot
\bar\jmath\geq0$, which yields
\begin{equation*}
  \L\super{\epsilon}(c,\bar\jmath) \geq {\Fi\as}\super{\epsilon}(c) + \tfrac12\grad\I\super{\epsilon}_0\cdot \Gamma\bar\jmath,
\end{equation*}
where the superscript $\epsilon$ denotes the functions corresponding to the modified rates $\kappa\super{\epsilon}$. Recall from
Theorem~\ref{th:time reversal} that the network with rates $\overleftarrow k\super{V}$ has the same invariant measure
$\pi\super{V}$ as the network with rates $k\super{V}$; by linearity the same is true for the network with rates
$k\super{V,\epsilon}$. It follows that the functional $\I\super{\epsilon}_0=\I_0$ remains unaltered. The Fisher information
${\Fi\as}\super{\epsilon}(c)$ is continuous in $\kappa\super{\epsilon}(c)$ and hence converges to $\Fi\as(c)$. Finally for the
cost function it follows by continuity that
\begin{align}
  \label{eq:approx FIR ineq}
  \Fi\as(c) + \tfrac12\grad\I_0(c)\cdot\Gamma\bar\jmath  & =\lim_{\epsilon\to0} {\Fi\as}\super{\epsilon}(c)
  + \tfrac12\grad\I_0\super{\epsilon}(c)\cdot\Gamma\bar\jmath\notag\\
  &\leq\inf_{j\in\RR^\R_+: j_r-j_{\bw(r)}=\bar\jmath_r}\, \limsup_{\epsilon\to0}\, \S\big(j \mid \kappa\super{\epsilon}(c)\big)\notag\\
  &=\inf_{j\in\RR^\R_+: j_r-j_{\bw(r)}=\bar\jmath_r}\,\S\big(j \mid \kappa(c)\big) =\L(c,\bar j).
\end{align}
Hence this inequality holds for any $(c,\bar j)$ for which $\kappa(c)\llgg\overleftarrow\kappa\!_{\bw}(c)$.

In the second step, take an arbitrary $(c,\bar\jmath)\in W^{1,1}(0,T;\RR^\Y_+)\times L^1(0,T;\RR^{\R_\fw})$ such that $\I_0$
is continuous in a neighbourhood of $c(0)$; without loss of generality we may assume that
$\int_0^T\!\L\big(c(t),\bar\jmath(t)\big)\,dt<\infty$. Recall from~\eqref{eq:FD} that
$\kappa(c)\llgg\overleftarrow\kappa\!_{\bw}(c)$ for almost every $c$. We distinguish between two cases. In the first case the path
$c(t)\equiv c$ is constant in time (and hence $\Gamma\bar\jmath=0$). Then we may approximate $c\super\epsilon\to c$ such that
$\kappa(c\super{\epsilon})\llgg\overleftarrow\kappa\!_{\bw}(c\super{\epsilon})$, hence~\eqref{eq:approx FIR ineq} holds for all
$t$, and also in time-integrated form. Then $\I_0(c\super{\epsilon}(T))-\I_0(c\super{\epsilon}(0))=0=\I_0(c(T))-\I_0(c(0))$, and
clearly the Fisher information converges. Then
\begin{align*}
  &\int_0^T\!\Fi\as\big(c(t)\big) + \tfrac12\I_0\big(c(T)\big) - \tfrac12\I_0\big(c(0)\big)\\
  &\qquad =\lim_{\epsilon\to0} \int_0^T\!\Fi\as\big(c\super{\epsilon}(t)\big) + \tfrac12\I_0\big(c\super{\epsilon}(T)\big)
  - \tfrac12\I_0\big(c\super{\epsilon}(0)\big)\\
  &\qquad\leq \limsup_{\epsilon\to0} \int_0^T\L\big(c\super{\epsilon},\bar\jmath(t)\big)\,dt \leq \inf_{\substack{j\in L^1(0,T;\RR^\R_+)\\
      j_r-j_{\bw(r)}=\bar\jmath_r}} \limsup_{\epsilon\to0} \int_0^T\!\S\big(j(t)\mid k(c\super{\epsilon})\big)\,dt\\
  &\qquad = \inf_{\substack{j\in L^1(0,T;\RR^\R_+)\\ j_r-j_{\bw(r)}=\bar\jmath_r}} \limsup_{\epsilon\to0}
  \int_0^T\!\big\lbrack\S\big(j(t)\mid k(c)\big) + \sum_{r\in\R} j_r(t)\log\frac{\kappa(c)}{\kappa(c\super{\epsilon})}\big\rbrack\,dt\\
  &\qquad =\int_0^T\L\big(c,\bar\jmath(t)\big)\,dt.
\end{align*}

In the second case the path $c(t)$ is not constant. Take any $j\in L^1(0,T;\RR_+^\R)$ for which
$j_r-j_{\bw(r)}=\bar\jmath_r$. Following~\cite[Lemmas~3.8 \& 3.9]{PattersonRenger2019} we add a little mass and convolute with a heat
kernel so that $c\super{\epsilon},j\super{\epsilon}\to c,j$ strongly in $L^1$-norm and $\int_0^T\!\S\big(j\super{\epsilon}(t)\mid
\kappa(c\super{\epsilon}(t))\big)\,dt\to \int_0^T\!\S\big(j(t)\mid \kappa(c(t))\big)\,dt$. Since $c(t)$ is not constant, the
convolved path $c\super{\epsilon}(t)$ only passes through points $c(t)$ that violate
$\kappa(c\super{\epsilon}(t))\llgg\overleftarrow\kappa\!_{\bw}(c\super{\epsilon}(t))$ on a $t$-null set. Hence
inequality~\eqref{eq:approx FIR ineq} holds for this pair $(c\super{\epsilon},\bar\jmath\super{\epsilon})$ and almost all $t$, and
hence also in integrated form. The convergence of the Fisher information follows by dominated convergence, where we recall from
Assumption~\ref{ass:kappa} that the rates $\kappa$ are bounded, and the reversed rates $\overleftarrow\kappa$ are bounded by
\eqref{eq:sum kfw is sum kbw}. The convergence of $\I_0(c\super{\epsilon}(T))$ is by lower semicontinuity and the convergence of
$\I_0(c\super{\epsilon}(0))$ by assumption. We conclude that the inequality~\eqref{eq:FIR} holds for any path as claimed.
\end{proof}

We can now also derive an inequality that bounds the other Fisher information together with the irreversible work.

\begin{proposition} 
  Assume that $\kappa$ satisfies the conditions of~\cite{PattersonRenger2019} needed for the large-deviations Theorem~\ref{th:ldp} to
  hold. Then for any arbitrary path $(c,\bar\jmath)\in W^{1,1}(0,T;\RR^\Y_+)\times L^1(0,T;\RR^\R)$ such that
  $\kappa_r(c(t))\llgg\overleftarrow\kappa_r(c(t))$ for almost all $t\in(0,T)$ and $r\in\R$,
\begin{equation*}
  \int_0^T\!\L\big(c(t),\bar\jmath(t)\big)\,dt \geq \int_0^T\!\Fi\sa\big(c(t)\big)\,dt - \int_0^T\!F\asym(c(t))\cdot\bar\jmath(t)\,dt.
\end{equation*}
\end{proposition}

\begin{proof} 
  This is the same argument as Proposition~\eqref{prop:FIR}, but starting from~\eqref{eq:orth split L asym} rather
  then~\eqref{eq:orth split L sym}. Note that the condition $\kappa_r(c(t))\llgg\overleftarrow\kappa\!_r(c(t))$ is now needed for
  $F\asym(c)$ to be well-defined.
\end{proof}

We finally remark that the irreversible work $\int_0^T\!F\asym(c(t))\cdot\bar\jmath(t)\,dt$ does not necessarily have a sign, but
in the other direction we have by convex duality
\begin{equation*}
  \int_0^T\! F\asym(c(t))\cdot \bar\jmath(t)\,dt \leq \int_0^T\!\Phi\big(c(t),\bar\jmath(t)\big)\,dt + 
  \int_0^T\!\Phi^*\big(c(t),F\asym(c(t))\big)\,dt.
\end{equation*}

\section*{Acknowledgements}

DRMR was funded by Deutsche Forschungsgemeinschaft (DFG) through grant CRC 1114 ``Scaling Cascades in Complex
Systems'', Project C08. JZ received funding through a Royal Society
Wolfson Research Merit Award.

\bibliographystyle{alpha}
\bibliography{library}

\newcommand{\etalchar}[1]{$^{#1}$}
\begin{thebibliography}{BDSG{\etalchar{+}}15}

\bibitem[ADPZ13]{Adams2013a}
Stefan Adams, Nicolas Dirr, Mark Peletier, and Johannes Zimmer.
\newblock Large deviations and gradient flows.
\newblock {\em Philos. Trans. R. Soc. Lond. Ser. A Math. Phys. Eng. Sci.},
  371(2005):20120341, 17, 2013.

\bibitem[AK11]{Anderson2011a}
D.F. Anderson and T.G. Kurtz.
\newblock Continuous time {M}arkov chain models for chemical reaction networks.
\newblock In H.~Koeppl, G.~Setti, M.~di~Bernardo, and D.~Densmore, editors,
  {\em Design and Analysis of Biomolecular Circuits}, pages 3--42, New York,
  N.Y. U.S.A, 2011. Springer.

\bibitem[BDSG{\etalchar{+}}15]{Bertini2015a}
Lorenzo Bertini, Alberto De~Sole, Davide Gabrielli, Giovanni Jona-Lasinio, and
  Claudio Landim.
\newblock Macroscopic fluctuation theory.
\newblock {\em Reviews of Modern Physics}, 87(2):593--636, 2015.

\bibitem[HPST19]{HilderPeletierSharmaTse2019}
B.~Hilder, M.A. Peletier, U.~Sharma, and O.~Tse.
\newblock An inequality connecting entropy distance, {F}isher information and
  large deviations.
\newblock {\em Stochastic Processes and their Applications}, online, 2019.

\bibitem[KJZ18]{Kaiser2018a}
M.~Kaiser, R.L. Jack, and J.~Zimmer.
\newblock Canonical structure and orthogonality of forces and currents in
  irreversible {M}arkov chains.
\newblock {\em Journal of Statistical Physics}, 170(6):1019--1050, 2018.

\bibitem[Kur70]{Kurtz1970a}
T.G. Kurtz.
\newblock Solutions of ordinary differential equations as limits of pure jump
  {M}arkov processes.
\newblock {\em Journal of Applied Probability}, 7(1):49--58, 1970.

\bibitem[Mae17]{Maes2017a}
C.~Maes.
\newblock Frenetic bounds on the entropy production.
\newblock {\em Phys. Rev. Lett.}, 119(16):160601, 2017.

\bibitem[MN08]{Maes2008a}
C.~Maes and K.~Neto{\v{c}}n{\'y}.
\newblock Canonical structure of dynamical fluctuations in mesoscopic
  nonequilibrium steady states.
\newblock {\em Europhys. Lett. EPL}, 82(3):Art. 30003, 6, 2008.

\bibitem[MPPR17]{Mielke2017a}
A.~Mielke, R.I.A. Patterson, M.A. Peletier, and D.R.M. Renger.
\newblock Non-equilibrium thermodynamical principles for chemical reactions
  with mass-action kinetics.
\newblock {\em SIAM Journal on Applied Mathematics}, 77(4):1562--1585, 2017.

\bibitem[MPR14]{Mielke2014a}
A.~Mielke, M.A. Peletier, and D.R.M. Renger.
\newblock On the relation between gradient flows and the large-deviation
  principle, with applications to {M}arkov chains and diffusion.
\newblock {\em Potential Analysis}, 41(4), 2014.

\bibitem[OM53]{Onsager1953a}
L.~Onsager and S.~Machlup.
\newblock Fluctuations and irreversible processes.
\newblock {\em Physical Rev. (2)}, 91(6):1505--1512, 1953.

\bibitem[PR19]{PattersonRenger2019}
R.I.A. Patterson and D.R.M. Renger.
\newblock Large deviations of jump process fluxes.
\newblock {\em Mathematical Physics, Analysis and Geometry}, 22(3):21, 2019.

\bibitem[Ren18a]{Renger2018a}
D.R.M. Renger.
\newblock Flux large deviations of independent and reacting particle systems,
  with implications for {M}acroscopic {F}luctuation {T}heory.
\newblock {\em Journal of Statistical Physics}, 172(5):1291--1326, 2018.

\bibitem[Ren18b]{Renger2018b}
D.R.M. Renger.
\newblock Gradient and {GENERIC} systems in the space of fluxes, applied to
  reacting particle systems.
\newblock {\em Entropy, special issue ``Probability Distributions and Maximum
  Entropy in Stochastic Chemical Reaction Networks''}, 20(8):596, 2018.

\bibitem[Sch76]{Schnakenberg1976a}
J.~Schnakenberg.
\newblock Network theory of microscopic and macroscopic behavior of master
  equation systems.
\newblock {\em Reviews of Modern Physics}, 48(4):571--585, 1976.

\end{thebibliography}

\end{document}